\documentclass{entcs} \usepackage{entcsmacro}
\usepackage{graphicx}

\usepackage{mathrsfs}

\def\lastname{Bigan, Steyaert, and Douady}
\begin{document}
\begin{frontmatter}
  \title{On necessary and sufficient conditions for proto-cell stationary growth} \author[X,P7]{E. Bigan
  \thanksref{ALL}\thanksref{myemail}}
  \author[X]{J.M. Steyaert\thanksref{coemail}}
  \author[P7]{S. Douady\thanksref{coemail2}}
  \address[X]{Laboratoire d'Informatique\\\'{E}cole Polytechnique\\
    Palaiseau, France}
  \address[P7]{Laboratoire Mati\`{e}re et Syst\`{e}mes Complexes\\ Universit\'{e} Paris Diderot\\
    Paris, France}    
    \thanks[ALL]{The authors would like to thank Lo\"{i}c Paulev\'{e}, Laurent Schwartz and Pierre Legrain for fruitful discussions.} \thanks[myemail]{Email:
    \href{mailto:erwan.bigan@m4x.org} {\texttt{\normalshape
        erwan.bigan@m4x.org}}}\thanks[coemail]{Email:
    \href{mailto:jean-marc.steyaert@polytechnique.edu} {\texttt{\normalshape
        jean-marc.steyaert@polytechnique.edu}}}\thanks[coemail2]{Email:
    \href{mailto:stephane.douady@univ-paris-diderot.fr} {\texttt{\normalshape
        stephane.douady@univ-paris-diderot.fr}}}
\begin{abstract} 
We consider a generic proto-cell model consisting of any conservative chemical reaction network embedded within a membrane. The membrane results from the self-assembly of one of the chemical species (membrane precursor) and is semi-permeable to some other chemical species (nutrients) diffusing from an outside growth medium into the proto-cell. Inside the proto-cell, nutrients are metabolized into all other chemical species including the membrane precursor, and the membrane grows in area and the proto-cell in volume. Investigating the conditions under which such a proto-cell may reach stationary growth, we prove that a simple necessary condition is that each moiety be fed with some nutrient flux; and that a sufficient condition for the existence of a stationary growth regime is that every siphon containing any species participating in the membrane precursor incorporation kinetics also contains the support of a moiety that is fed with some nutrient flux. These necessary and sufficient conditions hold regardless of chemical reaction kinetics, membrane parameters or nutrient flux diffusion characteristics.
\end{abstract}
\end{frontmatter}
\section{Introduction}\label{intro}
Proto-cell models are useful to help understand minimal requirements for the emergence of life~\cite{Kaneko,Kondo}. Models based on random complex chemical reaction networks (CRNs) have been developed to this end. Simpler autopoietic systems have also been proposed~\cite{LV}. All such models rely upon some CRN embedded inside a membrane that is synthesized from within. The membrane is semi-permeable to some nutrients flowing from an outside growth medium into the proto-cell. Owing to the embedded CRN, nutrients are metabolized into all other constituents, including the membrane precursor that gets incorporated into the membrane. The membrane thus grows in area and the proto-cell in volume, thereby diluting all constituents. A stationary growth regime was typically reached in previously proposed models~\cite{Kondo,LV,Bigan}, but it is unclear what grants such a property to the proto-cell. The present work addresses this question by proving necessary and sufficient conditions for such stationary proto-cell growth.

We introduce a generic proto-cell model consisting of any conservative chemical reaction network (CRN) embedded within a membrane. The membrane results from the self-assembly of one of the chemical species participating in the CRN (membrane precursor). This membrane is semi-permeable to some other chemical species (nutrient) diffusing across the membrane between an outside growth medium and the inside of the proto-cell. We prove a simple algebraic necessary condition as well as a topological sufficient condition for stationary proto-cell growth. Both conditions are independent of reaction kinetics, membrane characteristics or nutrient flux.

This paper is organized as follows. Section~\ref{model} presents the proto-cell model. Section~\ref{necsuffcond} gives the proof of the simple algebraic necessary condition as well as that of the sufficient topological condition. Section~\ref{discussion} is devoted to discussion. And Section~\ref{conclusion} gives a conclusion.

\section{Proto-cell model}
\label{model}

\subsection{Conservative CRNs}
We first define conservative CRNs and  their attributes. Following standard practice in Chemical Reaction Network Theory \cite{Aris,Erdi} a stoichiometry $N\times R$ matrix $S$ may be associated to any CRN where $N$ is the total number of chemical species $\mathscr{S}=\{A_{i}\}_{i=0,\dots,N-1}$ participating in the CRN and $R$ is the total number of reactions.
\begin{definition}
A CRN is \emph{conservative} if there exists a strictly positive $N\timesÊ1$ vector $\textbf{m}$ such that $\textbf{m}^{T}S=0$ where $(.)^{T}$ denotes the transpose of $(.)$. The $i^{th}$ coordinate of $\textbf{m}$, denoted as $m_{i}$, represents the molecular mass of species $A_{i}$.
\end{definition}
Basically, a CRN is conservative if each chemical species may be assigned a positive mass such that mass conservation be guaranteed for each chemical reaction. There may exist multiple solutions to the set of mass conservation equations. This notion is explicited through the definition below.
\begin{definition}
For a conservative CRN, the kernel of the transpose of $S$ denoted $Ker(S^{T})$ has dimension $dim(Ker(S^{T}))=p\geq1$, and the rank of $S$ is $dim(Im(S))=dim(Im(S^{T}))$ which is equal to $N-dim(Ker(S^{T}))=N-p\leq N-1$. The set of mass vectors compatible with mass conservation is $\{\textbf{m}~|~m_{i}>0,~i=0,\dots,N-1~and~\textbf{m}^{T}S=0\}$. This set constitutes a pointed convex cone having $p'\geq p$ generating vectors $\{\textbf{b}_{k}\}_{k=0,\dots, p'-1}$, $p$ of which are linearly independent and constitute a basis of $Ker(S^{T})$~\cite{SH}. We define as \emph{moieties} the elements of such a set of generating vectors $\{\textbf{b}_{k}\}_{k=0,\dots, p-1}$.
\end{definition}
Moieties essentially correspond to positive linear combinations of chemical species concentrations that are left invariant by the chemical reactions.
\begin{definition}
The \emph{support} of a moiety or positive linear combination of moieties $\textbf{b}$ denoted as $supp(\textbf{b})$ is the subset of chemical species $A_{i}$ along which $\textbf{b}$ has non-zero components.
\end{definition}

\subsection{Model assumptions}
The following assumptions are made in the proposed proto-cell model:
\begin{itemize}
\item Different chemistry inside and outside the cell: the CRN is active only inside the cell. This assumption is necessary for the cell to maintain itself out-of-equilibrium as the sole result of chemical reactions. It is verified in actual biological systems as the internal chemistry relies upon enzymes that undergo denaturation and/or that cannot be maintained at sufficiently high concentration outside the cell. The mechanisms through which such a different chemistry is maintained are beyond the scope of this paper.
\item Filament shape: as observed for ancestral biological systems such as bacteria and most archaea. With this assumption, the $Area/Volume$ ratio is kept constant when the cell is growing.
\item Self-assembly of one of the chemical species (membrane precursor $A_{me}$) in a structured membrane: the incorporation of $A_{me}$ into the growing membrane is assumed to be kinetically controlled, with rate $F_{me}$ per unit area, or $f_{me}$ per unit volume (the corresponding rate vector $\textbf{F}_{me}$, respectively $\textbf{f}_{me}$, has only one non-zero component: $F_{me}$, respectively $f_{me}$, along $A_{me}$). Owing to the previous assumption these rates per unit volume and per unit area are simply linked to one another through the constant $Area/Volume$ multiplicative factor: $f_{me}=F_{me}\times(Area/Volume)$. The membrane is further characterized by the number of molecules per unit area $N_{mea}$. Owing to the previous assumption, $N_{mea}$ can be multiplied by the constant $Area/Volume$ to result in an equivalent concentration: $C_{me}=N_{mea}\times(Area/Volume)$. This is the effective membrane concentration as if all self-assembled molecules were dissolved in the cell volume.
\item Membrane precursor incorporation kinetics: the membrane precursor incorporation rate $F_{me}$ per unit area (or $f_{me}$ per unit volume) is assumed to be a continuous monotonically increasing function of the concentration vector $\textbf{A}$ such that $F_{me}(\textbf{A}=\textbf{0})=0$ and $F_{me}(\textbf{A})>0$ iff $\textbf{A}$ has non-zero components along a subset $\mathscr{S}_{me}$ of $\mathscr{S}$. $\mathscr{S}_{me}$ includes at least $A_{me}$ (membrane precursor must be present in the cytoplasm for it to be incorporated in the membrane) and may also include other species (e.g. enzymes or other metabolites that may be required in case of catalyzed or active membrane precursor incorporation). This is a very mild assumption as it is verified by all foreseeable kinetics (mass-action, Michaelis-Menten, or active membrane precursor incorporation).
\item Semi-permeability of self-assembled membrane to a subset $\mathscr{S}_{nu}$ of $\mathscr{S}$ diffusing from an outside growth medium into the cell: the resulting nutrient flux vector is $\textbf{F}_{nu}$ per unit area and $\textbf{f}_{nu}$ per unit volume, with $\textbf{f}_{nu}=\textbf{F}_{nu}\times(Area/Volume)$. All proofs in the present work hold if $\textbf{f}_{nu}$ is either an independent control vector, or if it results from a saturating diffusion mechanism:

\begin{equation}
f_{nu,i}=f_{nu,i}(\lbrack A_{i}\rbrack)=\frac{\mathcal{D}_{i}(\lbrack A_{i,out}\rbrack-\lbrack A_{i}\rbrack)}{1+\frac{\lbrack A_{i,out}\rbrack}{K_{i,out}}} \hspace{4mm}\forall i\in\mathscr{S}_{nu}
\end{equation}

where for each nutrient species $A_{i}\in \mathscr{S}_{nu}$, $f_{nu,i}$ is the $\textbf{f}_{nu}$ component along $A_{i}$, $\mathcal{D}_{i}$ is the effective diffusion coefficient, $\lbrack A_{i,out}\rbrack$ is the fixed nutrient concentration in the outside growth medium, and $K_{i,out}$ is the saturation concentration.
\item Homogeneous concentrations: all chemical species are assumed to be homogeneously distributed and any intracellular diffusion effect is neglected. This is a simplifying assumption as one should expect nutrient $A_{i}\in \mathscr{S}_{nu}$ (resp. membrane precursor $A_{me}$) concentration to be highest (resp. lowest) near the membrane that acts as an effective source (resp. sink) for such chemical species. Similarly, any intracellular spatial organization is neglected.
\end{itemize}

Figure~\ref{protocell} gives a schematic representation of the proto-cell model.

\begin{figure}[!h]
   \includegraphics[scale=0.5]{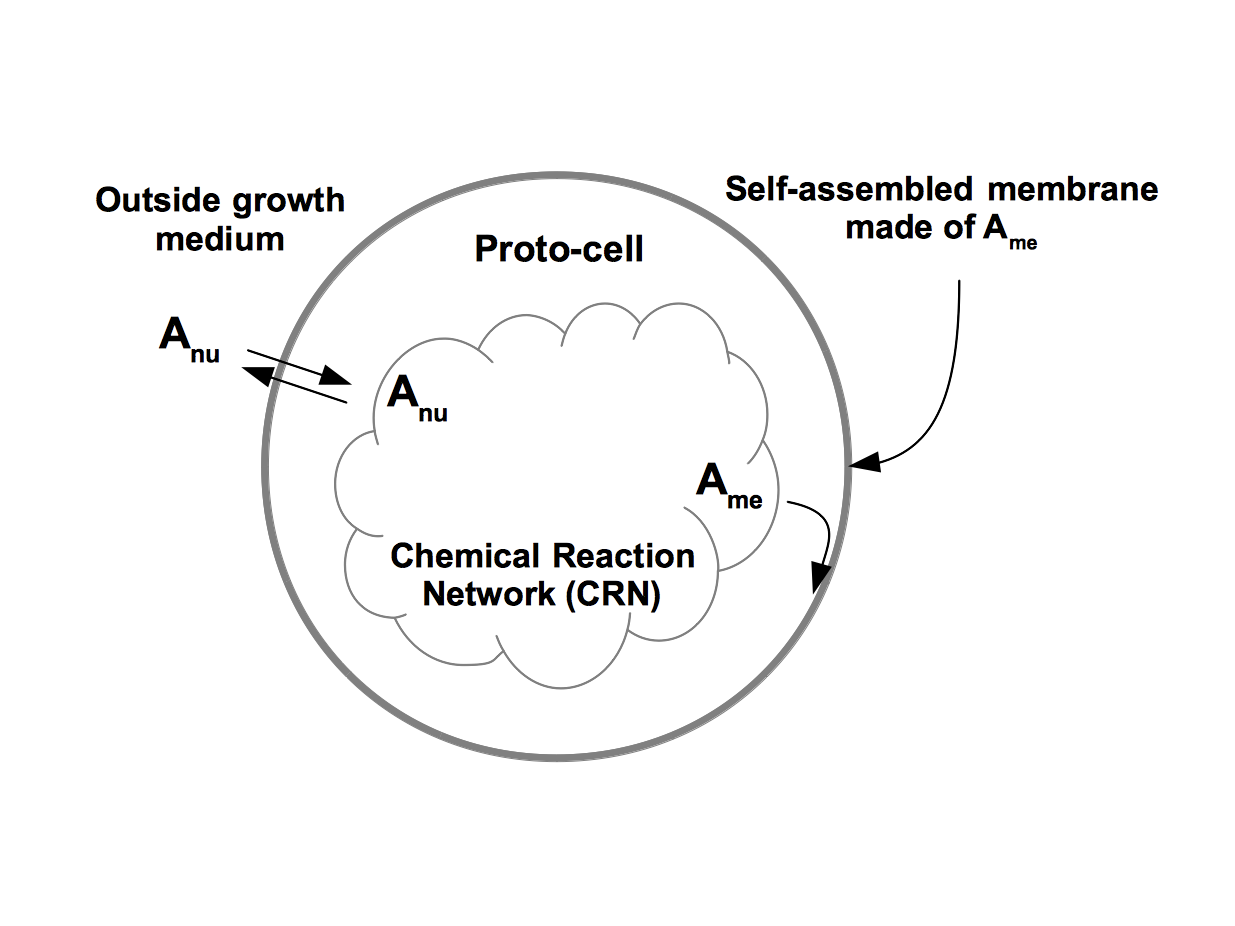}
   \caption{Schematic of the proto-cell model. The nutrient $A_{nu}$ flows bidirectionally across the membrane by diffusion. The membrane results from the self-assembly of the membrane precursor $A_{me}$. The membrane surface area grows as the result of $A_{me}$ unidirectional incorporation in the self-assembled membrane.}
   \label{protocell}
\end{figure}

\subsection{Ordinary Differential Equation (ODE) system}
With the above assumptions, the ODE system governing the time evolution of concentrations is given by:

\begin{equation}
\frac{d \textbf{A}}{dt}Ê=ÊS\textbf{f}+\textbf{f}_{nu}-\textbf{f}_{me}-\lambda\textbf{A}	\label{eq:ODE}
\end{equation}

where $S$ is the $N\times R$ stoichiometry matrix, $\textbf{A}$ is the $N\timesÊ1$ concentration vector having as components the $N$ species concentrations, $\textbf{f}$ is the $R\timesÊ1$ rate vector having as components the rates of each chemical reaction, $\textbf{f}_{nu}$ is the $N\timesÊ1$ nutrient flux vector with only non-zero components $f_{nu,i}$ for $A_{i} \in\mathscr{S}_{nu}$, $\textbf{f}_{me}$ is the $N\timesÊ1$ membrane incorporation flux vector with only non-zero component $f_{me}$ along $A_{me}$, and $\lambda\textbf{A}$ represents the dilution factor with growth rate $\lambda=(1/Volume)dVolume/dt$. From the second assumption, we have $\lambda=(1/Area)dArea/dt$ because the $Area/Volume$ ratio is constant. The latter expression for $\lambda$ can be related to $f_{me}$ and $C_{me}$ as follows. As $A_{me}$ gets incorporated into the growing membrane with rate (per unit area) $F_{me}$ per unit area, the membrane area grows as $(1/Area)dArea/dt=F_{me}/N_{mea}=f_{me}/C_{me}$. This gives:

\begin{equation}
\lambda=Ê\frac{f_{me}}{C_{me}}
\end{equation}

The above self-consistent set of equations thus gives the system of ODEs governing the time evolution of concentrations for the growing proto-cell. In the case where $\textbf{f}_{nu}$ is considered as a fixed control vector, the proto-cell ODE system is reminiscent of that for a Continuous Flow Stirred Tank Reactor (CFSTR)~\cite{Craciun}, a signicant difference being that the fixed ouput rate for a CFSTR is replaced by the variable growth rate $\lambda$ for the proto-cell.

\section{Necessary and sufficient conditions for stationary proto-cell growth}
\label{necsuffcond}

\subsection{Simple algebraic necessary condition}
\begin{theorem}
For a proto-cell to reach stationary growth, it is necessary that each moiety of the corresponding CRN be fed.
\end{theorem}
\begin{proof}
Multiplying both sides of the above ODE system by $\textbf{b}^{T}$ where  $\textbf{b}$ is any of the $p$ moieties $\{\textbf{b}_{k}\}_{k=0,\dots, p-1}$, gives the ODE governing the time evolution of the quantity $\textbf{b}^{T}\textbf{A}$:

\begin{equation}
\frac{d(\textbf{b}^{T}\textbf{A})}{dt}=\textbf{b}^{T}\textbf{f}_{nu}-\lambda(B_{me}+\textbf{b}^{T}\textbf{A})
\end{equation}

where $B_{me}=b_{me}C_{me}$, with $b_{me}$ being the component of $\textbf{b}$ along $A_{me}$. This shows that $\textbf{b}^{T}\textbf{f}_{nu}=0$ results either in an exponentially decreasing $\textbf{b}^{T}\textbf{A}\rightarrow 0$, or $\lambda\rightarrow 0$. In the former case all species asymptotically disappear, while in the latter case at least one of the species in $\mathscr{S}_{me}$ ($A_{me}$ or any other species the presence of which is required for precursor incorporation) asymptotically disappears. A necessary condition for stationary growth is thus that each moiety be fed: $\textbf{b}_{k}^{T}\textbf{f}_{nu}>0$ or equivalently $\mathscr{S}_{nu}\cap supp(\textbf{b}_{k}) \ne\varnothing$, for $k=0,\dots,p-1$.
\end{proof}

\subsection{Topological sufficient condition}
This sufficient condition builds upon the topological concept of siphon introduced by Angeli \emph{et al}~\cite{Angeli} to study persistence in CRNs. Siphons are essentially the sets of chemical species whose absence cannot be compensated by the chemical reactions that take place inside the cell, and may be formally defined as follows.
\begin{definition}
A \emph{siphon} $Z$ is a subset of $\mathscr{S}$ such that for each species $A_{i}$ in $Z$ and every reaction where $A_{i}$ appears as product, then at least one of the reactant species also belongs to $Z$.
\end{definition}
We shall see below that a sufficient condition for the proto-cell ODE system to admit a fixed point with positive growth rate is that every siphon containing any species in $\mathscr{S}_{me}$ also contains the support of a moiety $\textbf{b}$, and that this moiety be fed with some nutrient flux $\textbf{f}_{nu}$, $\textbf{b}^{T}\textbf{f}_{nu}>0$. The positive growth rate stems from the persistence of any species in $\mathscr{S}_{me}$. But this condition is not sufficient to guarantee persistence of all species. Further, if every siphon (not just those containing any species in $\mathscr{S}_{me}$) also contains the support of a moiety (and that this moiety is fed with some nutrient flux), then all chemical species in $\mathscr{S}$ are persistent and present in the fixed point. It should be noted that this condition only guarantees the existence of a fixed point, but does not guarantee its stability, nor its unicity. Numerical analyses on randomly generated CRNs suggest that this fixed point is actually stable and unique, but this remains a conjecture at this stage.

The proof of this topological sufficient condition builds upon the following series of lemmas.

\begin{lemma}
\label{Brouwer}
If the cytoplasmic density $D=\textbf{m}^{T}\textbf{A}$ remains bounded for the entire proto-cell trajectory given by the autonomous differential equation $d\textbf{A}/dt=g(\textbf{A})$ with $g(\textbf{A})=S\textbf{f}+\textbf{f}_{nu}-\textbf{f}_{me}-\lambda\textbf{A}$, then there exists a stationary point $\textbf{A}_{st}$ such that $g(\textbf{A}_{st})=0$.
\end{lemma}

\begin{proof}
The proof results from a generalization of Brouwer's fixed point theorem to dynamical systems, and was given by Wei~\cite{Wei}. This result is not specific to chemical systems and holds for any continuous autonomous dynamical system $d\textbf{x}/dt=g(\textbf{x})$ having a semiflow. If a subset K of the entire space of all possible $\textbf{x}$ is non-empty, convex and forward-invariant for the flow of $g$, then there exists a fixed point $\textbf{x}_{0}$ such that $g(\textbf{x}_{0})=0$. The same proof as that given by Wei was given in a different context~\cite{Basener}. A similar proof was also given in its most general context by Gnacadja~\cite{Gnacadja}.
\end{proof}

\begin{lemma}
\label{upperbound}
For any moiety or positive linear combination of moieties $\textbf{b}$ that is fed with some nutrient flux $\textbf{b}^{T}\textbf{f}_{nu}>0$ and such that all concentrations for chemical species in $\mathscr{S}_{me}$ admit strictly positive lower bounds, the corresponding quantity $\textbf{b}^{T}\textbf{A}$ has an upper bound. In particular, $D=\textbf{m}^{T}\textbf{A}$ is bounded.
\end{lemma}

\begin{proof}
If all concentrations for chemical species in $\mathscr{S}_{me}$ have stricly positive lower bounds, then so have the membrane precursor incorporation rate $f_{me}(t)\geq f_{me,min}>0$ and growth rate $\lambda(t)\geq\lambda_{min}=f_{me,min}/C_{me}>0$, and the time evolution of $\textbf{b}^{T}\textbf{A}(t)$ can be bounded as follows:

\begin{equation}
\frac{d(\textbf{b}^{T}\textbf{A})}{dt}\leq\textbf{b}^{T}\textbf{f}_{nu,max}-\lambda_{min}(B_{me}+\textbf{b}^{T}\textbf{A})
\end{equation}

where $\textbf{f}_{nu,max}$ denotes either the fixed nutrient flux $\textbf{f}_{nu}$ vector (if taken as independent control vector) or the vector having as component along $A_{i}\in\mathscr{S}_{nu}$ $f_{nu,max,i}=\mathcal{D}_{i}K_{i,out}$ (if nutrient flux results from a saturating diffusion mechanism). This bound holding for any $t\geq0$, $\textbf{b}^{T}\textbf{A}(t)$ is thus upper bounded by the function $F(t)$ solution of the differential equation obtained by replacing the above inequality by the equality, and having the same initial conditions as $\textbf{b}^{T}\textbf{A}(t)$:

\begin{equation}
\textbf{b}^{T}\textbf{A}(t)\leq F(t)
\end{equation}

with $F(t)$ solution of:

\begin{equation}
\frac{d F}{dt}=\textbf{b}^{T}\textbf{f}_{nu,max}-\lambda_{min}(B_{me}+F)
\end{equation}

with the same initial conditions as $\textbf{b}^{T}\textbf{A}$, $F(t=0)=\textbf{b}^{T}\textbf{A}(t=0)=B_{0}$. Solving analytically for this differential equation gives:

\begin{equation}
F(t)=\frac{\textbf{b}^{T}\textbf{f}_{nu,max}}{\lambda_{min}}\lbrack1-(1-\frac{(B_{0}+B_{me})\lambda_{min}}{\textbf{b}^{T}\textbf{f}_{nu,max}} )\mathrm{e}^{-\lambda_{min}t}\rbrack-B_{me}
\end{equation}

It can be seen that $F(t)$ asymptotically converges towards $F(tÊ=+\infty)=\frac{\textbf{b}^{T}\textbf{f}_{nu,max}}{\lambda_{min}}-B_{me}$ which is finite. This shows that $\textbf{b}^{T}\textbf{A}(t)$ is upper bounded.

Applying this to the particular case of the mass, $\textbf{b}=\textbf{m}$, shows that if all concentrations for chemical species in $\mathscr{S}_{me}$ have strictly positive lower bounds, then the cytoplasmic density $D$ is bounded.
\end{proof}

\begin{lemma}
\label{diffequ}
Consider a differential equation of the form $dy/dt=\psi(y)$, where $\psi$ is a continuous monotonically decreasing function of $y$ such that $\psi(y=0)>0$. With positive initial condition $y(t=0)>0$, the corresponding trajectory has a strictly positive lower bound, i.e. $\forall t\geq 0$, $y(t)\geq y_{min}>0$. With zero initial condition $y(t=0)=0$, the corresponding trajectory has a strictly positive lower bound for times beyond any $t_{1}>0$, i.e. $\forall t_{1}>0$, $\exists y_{min}>0$ such that $\forall t\geq t_{1}$, $y(t)\geq y_{min}$.
\end{lemma}

\begin{proof}
As $\psi(y)$ is a continuous monotonically decreasing function having a positive value at $y=0$, either it converges towards some $\psi(y=+\infty)\geq0$, or it admits a single zero $y_{0}>0$, $\psi(y=y_{0})=0$. In the former case, $y(t)$ increases at all times and is thus lower bounded by its initial condition if positive $y(t=0)>0$, or by $y(t=t_{1})>0$ for any $t_{1}>0$ in case of zero initial condition. In the latter case, we shall prove that $y(t)$ converges asymptotically towards the zero $y_{0}$ of $\psi$. To this end, we define as $\delta$ the square of the distance between $y(t)$ and $y_{0}$:

\begin{equation}
\delta(t)=(y(t)-y_{0})^{2}
\end{equation}

Its derivative is given by:

\begin{equation}
\frac{d \delta}{dt}=2(y(t)-y_{0})\times dy/dt=2(y(t)-y_{0})\times\psi(y)
\end{equation}

$d\delta/dt\leq0$ because it is the product of two quantities having opposite signs. $\delta(t)$ is thus positive and monotonically decreasing, and thus converges towards some finite $\delta(t=+\infty)\geq 0$ for which $d\delta/dt(t=+\infty)=0$. This implies $y(t=+\infty)=y_{0}$. $y(t)$ converges asymptotically towards $y_{0}>0$ with a distance to this limit that is monotonically decreasing. In case of positive initial condition, $y(t)$ is thus always lower bounded by the smallest of $y_{0}$ and $y(t=0)$, both being positive. In case of zero initial condition, beyond any time $t_{1}>0$, $y(t)$ is lower bounded by the smallest of $y_{0}$ and $y(t=t_{1})$, both being positive.
\end{proof}

\begin{lemma}
\label{lowerbound}
For any moiety or positive linear combination of moieties $\textbf{b}$ such that $\mathscr{S}_{me}\subset supp(\textbf{b})$ and such that $\textbf{b}$ is fed with some nutrient flux (i.e. $\textbf{b}^{T}\textbf{f}_{nu}>0$ or equivalently $\mathscr{S}_{nu}\cap supp(\textbf{b}) \ne\varnothing$), then the corresponding quantity $\textbf{b}^{T}\textbf{A}$ has a strictly positive lower bound, for all $t\geq 0$ in case of positive initial conditions $\textbf{b}^{T}\textbf{A}(t=0)>0$, or for all $t\geq t_{1}>0$ in case of zero initial conditions $\textbf{b}^{T}\textbf{A}(t=0)=0$.
\end{lemma}

\begin{proof}
We first consider the case where $f_{nu}$ is an independent control parameter. The growth rate $\lambda$ can be bounded by a function of the quantity $\textbf{b}^{T}\textbf{A}$:

\begin{equation}
\lambda=\frac{f_{me}(\{\lbrack A_{i}\rbrack\}_{i \in \mathscr{S}_{me}})}{C_{me}}=\frac{f_{me}(\{\frac{b_{i}\lbrack A_{i}\rbrack}{b_{i}}\}_{i \in \mathscr{S}_{me}})}{C_{me}}\leq\frac{f_{me}(\{\frac{\textbf{b}^{T}\textbf{A}}{b_{i}}\}_{i \in \mathscr{S}_{me}})}{C_{me}}
\end{equation}

The latter inequality results from $f_{me}$ being a monotonically increasing function of all concentrations for species in $\mathscr{S}_{me}$. Feeding this upper bound for $\lambda$ in the dynamical equation governing the evolution of $\textbf{b}^{T}\textbf{A}$ provides a lower bound for $d(\textbf{b}^{T}\textbf{A})/d$t:

\begin{equation}
\frac{d(\textbf{b}^{T}\textbf{A})}{dt}\geq\textbf{b}^{T}\textbf{f}_{nu}-\frac{f_{me}(\{\frac{\textbf{b}^{T}\textbf{A}}{b_{i}}\}_{i \in \mathscr{S}_{me}})}{C_{me}}(B_{me}+\textbf{b}^{T}\textbf{A})
\end{equation}

This bound holding for any $t\geq0$, $\textbf{b}^{T}\textbf{A}(t)$ is thus lower bounded by the function $E(t)$ solution of the differential equation obtained by replacing the above inequality by the equality, and having the same initial conditions as $\textbf{b}^{T}\textbf{A}(t)$:

\begin{equation}
\textbf{b}^{T}\textbf{A}(t)\geq E(t)
\end{equation}

with E(t) solution of:

\begin{equation}
dE/dt=\textbf{b}^{T}\textbf{f}_{nu}-\frac{\phi(E)}{C_{me}}(B_{me}+E)
\end{equation}

where:

\begin{equation}
\phi(E)=f_{me}(\{\frac{E}{b_{i}}\}_{i \in \mathscr{S}_{me}})
\end{equation}

with the same initial conditions as $\textbf{b}^{T}\textbf{A}$, $E(t=0)=\textbf{b}^{T}\textbf{A}(t=0)=B_{0}$. The right hand side ($RHS$) of the above differential equation is a continuous monotonically decreasing function $RHS(E)$, with $RHS(E=0)=\textbf{b}^{T}\textbf{f}_{nu}>0$. It follows from Lemma~\ref{diffequ} that $E(t)$ is lower bounded, and so is $\textbf{b}^{T}\textbf{A}(t)$.

The same reasoning holds for the case where $\textbf{f}_{nu}$ results from a saturating diffusion mechanism. For each $A_{i}$ in $\mathscr{S}_{nu}$, the quantity $b_{i}f_{nu,i}$ is proportional to $\lbrack A_{i,out}\rbrack-\lbrack A_{i}\rbrack$. The sum of these contributions thus takes the form:

\begin{equation}
\textbf{b}^{T}\textbf{f}_{nu}=p-\sum_{i\in\mathscr{S}_{nu}}q_{i}\lbrack A_{i}\rbrack
\end{equation}

where $p$ and $\{q_{i}\}_{i\in\mathscr{S}_{nu}}$ are positive constants (that depend on outside nutrient concentrations and on diffusion parameters, but that are independent of cytoplasmic concentrations). This can be upper bounded as follows:

\begin{equation}
\textbf{b}^{T}\textbf{f}_{nu}=p-\sum_{i\in\mathscr{S}_{nu}}q_{i}\frac{b_{i}\lbrack A_{i}\rbrack}{b_{i}}\geq p-\sum_{i\in\mathscr{S}_{nu}}q_{i}\frac{\textbf{b}^{T}\textbf{A}}{b_{i}}=p-r\textbf{b}^{T}\textbf{A}
\end{equation}

Feeding this upper bound in the previous expression for the time derivative of $\textbf{b}^{T}\textbf{A}$ gives:

\begin{equation}
\frac{d(\textbf{b}^{T}\textbf{A})}{dt}\geq p-r\textbf{b}^{T}\textbf{A}-\frac{f_{me}(\{\frac{\textbf{b}^{T}\textbf{A}}{b_{i}}\}_{i \in \mathscr{S}_{me}})}{C_{me}}(B_{me}+\textbf{b}^{T}\textbf{A})
\end{equation}

This bound holding for any $t\geq0$, $\textbf{b}^{T}\textbf{A}(t)$ is thus lower bounded by the function $E(t)$ solution of the differential equation obtained by replacing the above inequality by the equality, and having the same initial conditions as $\textbf{b}^{T}\textbf{A}(t)$:

\begin{equation}
\textbf{b}^{T}\textbf{A}(t)\geq E(t)
\end{equation}

with E(t) solution of:

\begin{equation}
\frac{d E}{dt}=p-rE-\frac{\phi(E)}{C_{me}}(B_{me}+E)
\end{equation}

with the function $\phi$ as defined above, and with the same initial conditions as $\textbf{b}^{T}\textbf{A}$, $E(t=0)=\textbf{b}^{T}\textbf{A}(t=0)=B_{0}$. Here again, the right hand side ($RHS$) of the above differential equation is also a continuous monotonically decreasing function $RHS(E)$, with $RHS(E=0)=p>0$. It follows from Lemma~\ref{diffequ} that $E(t)$ is lower bounded, and so is $\textbf{b}^{T}\textbf{A}(t)$.
\end{proof}

\begin{lemma}
\label{siphon}
If every siphon containing any species in $\mathscr{S}_{me}$ contains the support of a moiety, and if this moiety is fed, then the concentrations of all species in $\mathscr{S}_{me}$ have strictly positive lower bounds for all times after a sufficient time (i.e. are persistent).
\end{lemma}
\begin{proof}
It is based upon the previous work of Angeli \emph{et al}~\cite{Angeli} that gives a sufficient condition for persistence (of any species) in a chemical system. This sufficient condition is that any siphon includes the support of a moiety. Their proof relies upon a key result (Proposition~5.4. in~\cite{Angeli}) which we restate here in a slightly different wording:
\begin{proposition}
\label{Angeliprop}
If for some particular set of strictly positive initial conditions, the corresponding $\mathrm{\omega}$-limit set has zero concentration for some non-empty subset $Z$ of all species, then $Z$ is a siphon.
\end{proposition}
The proof of their sufficient condition for persistence then proceeds as follows~\cite{Angeli}: if the system were not persistent, then there would exist (for some particular set of strictly positive initial concentration vector $\xi$), some non-empty subset $Z$ of species and a suitable time sequence $\{t_{n}\}$ such that for any $A_{i}$ in $Z$, $\lbrack A_{i}\rbrack(t_{n})\rightarrow0$ at least along the suitable sequence $t_{n}\rightarrow+\infty$. And by virtue of Proposition~\ref{Angeliprop}, $Z$ would be a siphon. Since any siphon contains the support of a moiety, say $\textbf{b}$, then $\textbf{b}^{T}\textbf{A}$ would be such that $\textbf{b}^{T}\textbf{A}(t)\rightarrow0$ if $t_{n}\rightarrow+\infty$. But this contradicts the fact that $\textbf{b}^{T}\textbf{A}(t)=\textbf{b}^{T}\xi>0$.

Following the same line of thought for the proto-cell, if some species $A_{i}\in\mathscr{S}_{me}$ were not persistent, there would then exist some non-empty species subset $Z$ (including $A_{i}$) and a suitable time sequence $\{t_{n}\}$ such that for any $A_{j}$ in $Z$, $\lbrack A_{j}\rbrack(t_{n})\rightarrow0$ at least along the suitable sequence $t_{n}\rightarrow+\infty$. And by virtue of Proposition~\ref{Angeliprop}, $Z$ would be a siphon. Since any siphon $Z$ containing $A_{i}\in\mathscr{S}_{me}$ contains the support of a moiety $\textbf{b}$ that is being fed, $\textbf{b}^{T}\textbf{f}_{nu}>0$, then $\textbf{b}^{T}\textbf{A}$ would be such that $\textbf{b}^{T}\textbf{A}(t)\rightarrow0$ if $t_{n}\rightarrow+\infty$. But this contradicts Lemma~\ref{lowerbound} which shows that $\textbf{b}^{T}\textbf{A}(t)$ has a strictly positive lower bound.
\end{proof}

We now combine these four lemmas to prove the topological sufficient condition for proto-cell stationary growth.

\begin{theorem}
If every siphon containing any species in $\mathscr{S}_{me}$ also contains the support of a moiety $\textbf{b}$, and if this moiety is fed with some nutrient flux $\textbf{f}_{nu}$, $\textbf{b}^{T}\textbf{f}_{nu}>0$, then there exists a fixed point with all species in $\mathscr{S}_{me}$ being persistent and with positive growth rate for the proto-cell ODE system. Further, if every siphon (not just those containing any species in $\mathscr{S}_{me}$) also contains the support of a moiety (and that this moiety is fed with some nutrient flux), then all chemical species are persistent and present in the fixed point.
\end{theorem}
\begin{proof}
If every siphon containing any species in $\mathscr{S}_{me}$ also contains the support of a moiety $\textbf{b}$, and if this moiety is fed with some nutrient flux $\textbf{f}_{nu}$, $\textbf{b}^{T}\textbf{f}_{nu}>0$, Lemma~\ref{siphon} ensures that all species in $\mathscr{S}_{me}$ are persistent. Applying Lemma~\ref{upperbound} to the particular case of the mass ensures that if all species in $\mathscr{S}_{me}$ are persistent, then the cytoplasmic density $D$ is bounded. And Lemma~\ref{Brouwer} ensures that if $D$ is bounded, the proto-cell dynamical system has a stationary point. Persistence of all species in $\mathscr{S}_{me}$ further ensures that this stationary point corresponds to a positive growth rate.

Further, if every siphon (not just those containing any species in $\mathscr{S}_{me}$) also contains the support of a moiety (and if this moiety is fed with some nutrient flux), then Lemma~\ref{siphon} ensures that all chemical species are persistent and present in the fixed point.
\end{proof}

\section{Discussion}
\label{discussion}
\subsection{Towards a stronger necessary condition}
Numerical analyses on randomly generated conservative CRNs suggest that the sufficient condition may also be necessary in the following sense: if there exists some siphon $Z$ containing some species in $\mathscr{S}_{me}$ that is not fed (while still having every moiety fed), it appears possible to find some system parameters (chemical reaction kinetics, membrane precursor incorporation kinetics, equivalent membrane concentration) for which no stationary proto-cell growth is possible. However, without any formal proof, this remains a conjecture at this stage.

\subsection{Production, degradation, and leakage of chemical species}
The proposed proto-cell model relies upon any conservative CRN. This also encompasses any situation where elementary building blocks (e.g. nucleotides) may be degraded or produced from simpler constituents inside the proto-cell.

The proved necessary and sufficient conditions even hold in the case where some species may leak out of the membrane. This is equivalent to having the membrane semi-permeable to a larger set of chemical species (nutrient species and leaking species), with only the nutrient species being present in the outside growth medium. The same necessary and sufficient conditions grant the existence of a stationary growth regime. In such a regime, stationary concentrations settle to values such that there is a net positive material influx to ensure a positive growth rate, i.e. the rate of mass nutrient influx exceeds the rate of mass leakage.

\subsection{Significance of the present work}
The sufficient condition for stationary growth is verified by a broad range of networks. Generating numerically large random conservative CRNs with arbitrary stoichiometry, we have found that once the CRN reaches a certain number of reactions, the quasi-totality of these networks are such that the only siphon is the full set of species, so that the sufficient condition is verified for any choice of membrane precursor $A_{me}$,  $\mathscr{S}_{me}$, or nutrients $\mathscr{S}_{nu}$. It could thus appear that a broad range of chemical reaction networks are compatible with the emergence of life.

However, thermodynamical constraints have not been taken into consideration in the present work. Further endowing random conservative CRNs with thermodynamically-consistent kinetics, we have found that the nutrient (resp. membrane precursor) should tend to have a high (resp. low) standard free Gibbs energy of formation and a low (resp. high) molecular weight for the proto-cell to exhibit a significant growth rate.~\cite{Bigan2}

There is another aspect of the model for which assumptions are actually quite stringent. These are assumptions made with respect to the membrane: self-assembly and semi-permeability to some nutrient, not to mention the importance of the self-assembly shape (filament vs. spherical) discussed above. This is in line with previous work that has stressed the importance of membranes in the origins of life.~\cite{Luisi}

The proposed model is applicable to proto-cells and not to current evolved unicellular organisms because the nutrient transport mechanism is passive diffusion for the proto-cell whereas it is active transport for current evolved cells. It can be verified that the necessary condition still holds in the case of active transport, but that the proposed proof for the sufficient condition does not. This is because the positive lower bound in Lemma~\ref{lowerbound} can no longer be guaranteed in the case of active transport as the nutrient influx rate goes to zero when all concentrations inside the cell go to zero.

\section{Conclusion}
\label{conclusion}
We have proposed a generic proto-cell model consisting of any chemical reaction network embedded within a membrane resulting from the self-assembly of one the chemical species participating in the reaction network. The proto-cell is assumed to have a filament shape and its membrane is assumed to be semi-permeable to some other chemical species (nutrients) diffusing from an outside growth medium into the cytoplasm. With these assumptions, we have proved necessary and sufficient conditions for proto-cell stationary growth. The necessary condition is purely algebraic and states that each moiety (corresponding to a conserved quantity when the chemical reaction network is closed) must be fed with some nutrient. The sufficient condition is topological and states that every siphon containing any species required for the membrane precursor incorporation contains the support of a moiety (and that this moiety is fed). The range of networks satisfying these conditions is potentially broad, which suggests that the most stringent requirements for the emergence of life may lie with the membrane, details of the embedded chemical reaction network being only of secondary importance.

\end{document}